\documentclass[conference]{IEEEtran}
\IEEEoverridecommandlockouts
\usepackage{cite}
\usepackage{easybmat}
\usepackage{stfloats}
\usepackage{amsmath,amssymb,amsfonts}
\usepackage{graphicx}
\usepackage{textcomp}
\usepackage[linesnumbered,lined,boxed,commentsnumbered, ruled]{algorithm2e}
\usepackage{epstopdf}
\def\BibTeX{{\rm B\kern-.05em{\sc i\kern-.025em b}\kern-.08em
    T\kern-.1667em\lower.7ex\hbox{E}\kern-.125emX}}

\usepackage{caption}
\usepackage{subcaption}

\newtheorem{thm}{Theorem}

\newtheorem{lem}{Lemma}
\newtheorem{defi}{Definition}
\newtheorem{ex}{Example}

\newenvironment{proof}{ \paragraph*{\hspace{-1em}Proof}}{\hfill$\square$}
 
\makeatletter
\def\widebre{\mathpalette\wide@breve}
\def\wide@breve#1#2{\sbox\z@{$#1#2$}%
     \mathop{\vbox{\m@th\ialign{##\crcr
\kern0.08em\brevefill#1{0.8\wd\z@}\crcr\noalign{\nointerlineskip}%
                    $\hss#1#2\hss$\crcr}}}\limits}
\def\brevefill#1#2{$\m@th\sbox\tw@{$#1($}%
  \hss\resizebox{#2}{\wd\tw@}{\rotatebox[origin=c]{90}{\upshape(}}\hss$}
\makeatletter

\begin{document}

\title{Trade-offs In Quasi-Decentralized Massive MIMO}
\author{Juan Vidal Alegr\'{i}a\IEEEauthorrefmark{1}, Fredrik Rusek\IEEEauthorrefmark{1}\IEEEauthorrefmark{2}, Jes\'{u}s Rodr\'{i}guez S\'{a}nchez\IEEEauthorrefmark{1}, Ove Edfors\IEEEauthorrefmark{1} \\
\IEEEauthorblockA{\IEEEauthorrefmark{1}Department of Electrical and Information Technology, Lund University, Lund, Sweden\\} 
\IEEEauthorblockA{\IEEEauthorrefmark{2}Sony Research Center, Lund, Sweden\\} 
\{juan.vidal\_alegria@eit.lth.se\}
\thanks{This work has been financed by the Swedish research council (VR) through grant 2017-04829 and by the strategic research program ELLIIT.}
}
\maketitle

\begin{abstract}
Typical massive multiple-input multiple-output (MIMO) architectures consider a centralized approach, in which all baseband data received by each antenna has to be sent to a central processing unit (CPU) to be processed. Due to the enormous amount of antennas expected in massive MIMO base stations (BSs), the number of connections to the CPU required in centralized massive MIMO is not scalable. In recent literature decentralized approaches have been proposed to reduce the number of connections between the antennas and the CPU. However, the reduction in the connections to the CPU requires more outputs per antenna to be generated. We study the trade-off between number of connections to the CPU and number of outputs per antenna. We propose a generalized architecture that allows exploitation of this trade-off, and we define a novel matrix decomposition that allows lossless linear equalization within our proposed architecture.

\end{abstract}


\section{Introduction}
\label{section:intro}
Massive multiple-input multiple-output (MIMO) is a key technology for the new generation of wireless communication systems \cite{marzetta,rusek}. By equipping the base station (BS) with hundreds of antennas, the spacial resolution is considerably increased, providing improved spectral efficiency since a large number of users can be spatially multiplexed and served in the same time-frequency resource.

Practical implementations of massive MIMO \cite{lumami,argos,bigstation} have already confirmed the benefits of massive MIMO. However, they also show that, in order to carry out centralized processing, the interconnection bandwidth between the antennas and the central processing unit (CPU) is prohibitively high since it linearly scales with the number of antennas. For this reason, a recent trend towards more decentralized architectures has emerged \cite{cavallaro,larsson,isit_2019,muris,jesus,vtc}. The general idea is to pre-process the baseband raw data at the antenna end to reduce the transfer of information to the CPU and make it practically scalable as the number of antennas grows large.

Available literature on decentralized massive MIMO proposes a wide range of solutions from fully decentralized, where channel state information (CSI) doesn't have to be available at the CPU \cite{muris, jesus,vtc,feedforward,cavallaro}, to partially decentralized architectures, where some of the processing tasks are distributed, but either full \cite{larsson} or partial CSI \cite{isit_2019} is available at the CPU. 

In \cite{isit_2019} it is argued that an architecture is decentralized enough if it doesn't need extra hardware apart from the minimum required during the payload data phase. It also states that the volume of data transferred during the data phase has to be independent of the number of antennas at the base station. However, as happens in \cite{isit_2019,muris,jesus,vtc}, in order reduce this volume of data, i.e., the number connections between the BS and the CPU, and make it independent of the number of antennas, each antenna has to provide a number of outputs that scales with the number of users (in a centralized architecture we would have only one output per antenna). We notice the existence of a trade-off between the number of connections to the CPU and the number of outputs from each antenna which is worth studying. We propose a generalized architecture that admits a trade-off between number of connections to the CPU and number of outputs per antenna. We define a matrix decomposition that allows the exploitation of this trade-off without any loss in capacity compared to centralized linear equalization.

The rest of the paper is organized as follows. The system model with our proposed generalized architecture is presented in Section \ref{section:model}. In Section \ref{section:wax} we present a novel matrix decomposition that is useful when performing lossless linear equalization within our generalized architecture. Section \ref{section:tradeoff} presents a discussion of the resulting trade-off together with some numerical examples of the defined matrix decomposition. We conclude the paper in Section \ref{section:conclusions} with some final remarks and future work.

Notation: In this paper, lowercase, bold lowercase and bold uppercase
letters stand for scalars, column vectors and matrices, respectively. The
operations $(.)^T$, $(.)^*$ and $(.)^H$ denote transpose, conjugate and conjugate transpose, respectively. The operation $\mathrm{diag}(.)$ gives a block diagonal matrix with the input matrices/vectors as the diagonal blocks. $\mathbf{I}_i$ corresponds to the identity matrix of size $i$, $\boldsymbol{1}_{i\times j}$ denotes the $i \times j$ all ones matrix, and $\boldsymbol{0}_{i \times j}$ denotes the $i \times j$ all zeros matrix.

\section{System model}
\label{section:model}
Let us consider $K$ single-antenna users transmitting to a base-station (BS) with $M$ antennas through a narrow-band channel with IID Rayleigh fading. The $M\times 1$ received complex vector, $\boldsymbol{y}$, can be expressed as
\begin{equation}
\boldsymbol{y} = \boldsymbol{H}\boldsymbol{s} + \boldsymbol{n},
\label{eq:ul_model}
\end{equation}
where $\boldsymbol{H}$ is the $M\times K$ channel matrix, $\boldsymbol{s}$ is the $K \times 1$ vector of symbols transmitted by the users, and $\boldsymbol{n}$ is a zero-mean complex white Gaussian noise vector with sample variance $N_0$. Considering a massive MIMO system we can assume that $M \gg K$.

The front-end of the receiver can, without loss of generality, be selected as spatially-matched filter (MF),
\begin{equation}
\boldsymbol{z} = \boldsymbol{H}^H\boldsymbol{y}.
\end{equation}
MF is lossless in the sense that it maintains the mutual information with respect to the transmitted vector $\boldsymbol{s}$, i.e., it achieves maximum capacity if optimum processing is applied over $\boldsymbol{z}$. Furthermore, linear equalizers, such as the MMSE and the zero-forcing ones, can be implemented based on $\boldsymbol{z}$. For simplicity, we only consider the problem of applying MF in a quasi-decentralized way, where our definition of quasi-decentralized relates only to the number of connections to the CPU required during the data phase.

In a typical centralized massive MIMO scenario with linear equalization, during the data phase, each antenna has to send one complex value to the CPU, corresponding to the entry of $\boldsymbol{y}$ received by that antenna ($y_m$). Therefore, a total of $M$ connections to the CPU are required, which can pose a scalability problem in a massive MIMO scenario where $M$ is expected to be large. The CPU, which has access to full channel state information (CSI) acquired during the training phase, is in charge of applying MF or any other linear equalization. In Figure \ref{fig:arch_old} (left) a schematic of this architecture is shown.

The decentralized massive MIMO architecture presented in \cite{isit_2019} reduces the number of connections to the CPU to $K$, which improves its scalability in the number of antennas. However, each antenna now has to transmit a vector of size $K \times 1$ since it has to multiply its received signal, $y_m$, by the conjugate of the local channel vector, $\boldsymbol{h}^H_m$, before forwarding it to the adding module that combines the antenna outputs and send the resulting $K \times 1$ vector to the CPU. In Figure \ref{fig:arch_old} (right) a schematic of this architecture  is shown. Note that, in this case, the CPU already receives a pre-processed signal, corresponding to the received vector after applying MF, i.e., $\boldsymbol{z}$, and the CPU can perform further processing of this signal.

\begin{figure}[h]
	\centering
		\vspace*{-5mm}
	\includegraphics[scale=0.485]{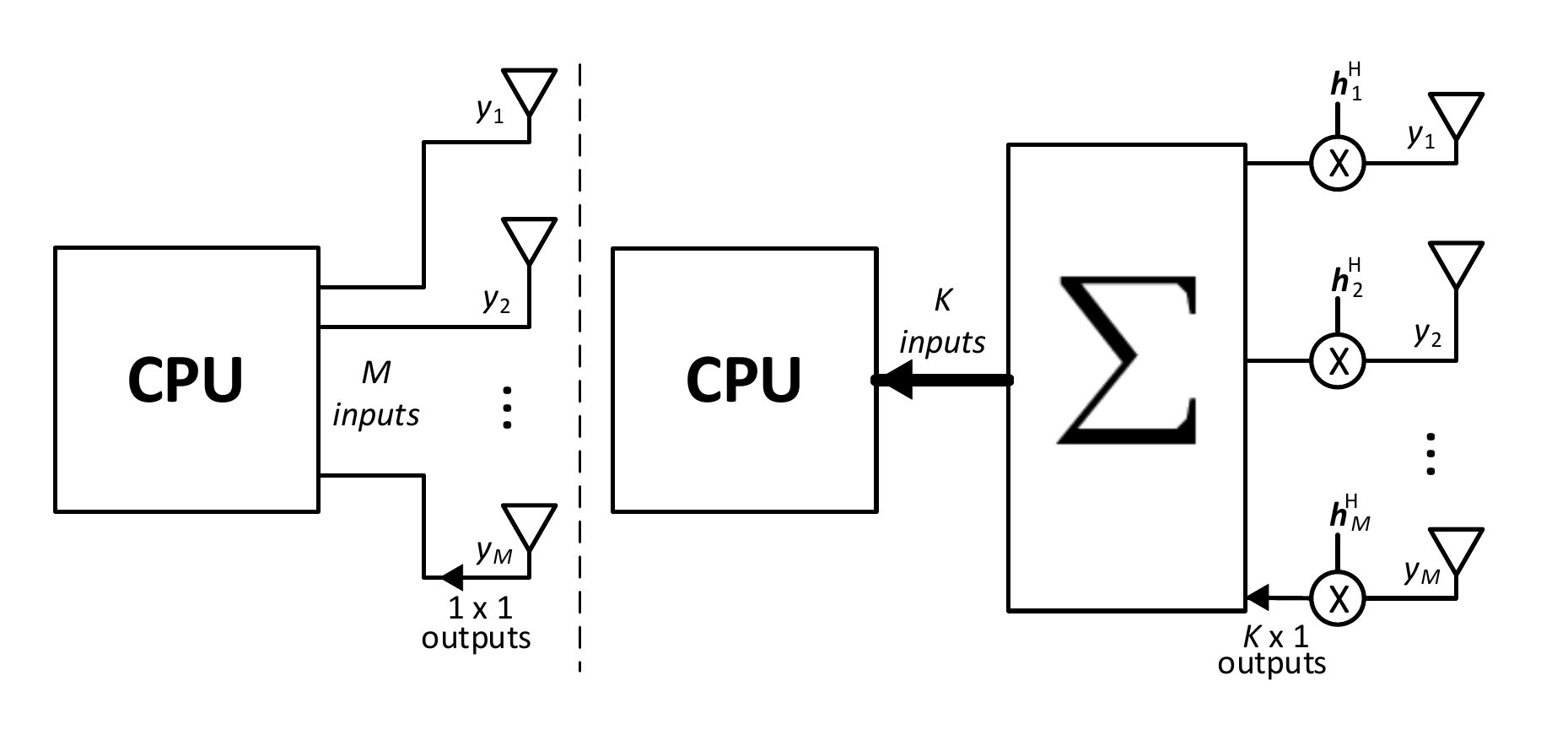}
		\vspace*{-4mm}
\caption{Architecture of centralized massive MIMO (left), and decentralized massive MIMO from \cite{isit_2019} (right).}
	\label{fig:arch_old}
\end{figure}

Let us consider the number of inputs to the CPU, $T$, and the number of outputs per antenna, $L$, as two important design parameters for a massive MIMO system. The values of these two parameters within the two architectures depicted in Figure \ref{fig:arch_old} suggest that there exists a trade-off between these two parameters, as shown in Figure \ref{fig:outputs}, where each of the architectures gives one point in the trade-off between these two parameters. However, it is unclear if a lossless trade-off, maintaining the full MF performance, between the two end-points exists and what it would look like.

\begin{figure}[h]
	\centering
	\includegraphics[scale=0.9]{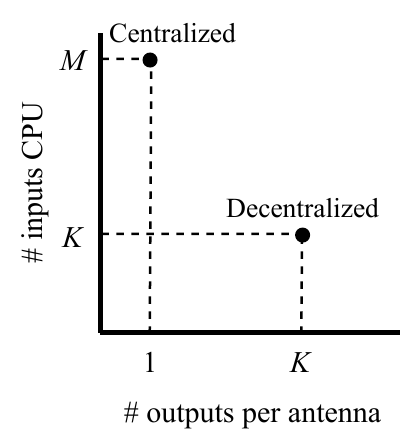}
\caption{Number of inputs to the CPU v.s. number of outputs per antenna for the centralized and decentralized massive MIMO architectures.}
	\label{fig:outputs}
\end{figure}

We propose a generalized architecture, depicted in Figure \ref{fig:arch_new}, that allows to contemplate any combination of $L$ and $T$, including the cases corresponding to the architectures from Figure \ref{fig:arch_old}. Note that, although any combination of $L$ and $T$ is possible within this architecture, it is not obvious whether a certain combination is able to achieve perfect MF. As can be seen in Figure \ref{fig:arch_new}, each antenna multiplies the received scalar value by an $L\times 1$ vector, $\boldsymbol{w}_m^H$, and the outputs are combined through a $T \times ML$ fixed linear operator\footnote[1]{The use of Hermitian is because the front-end receiver we want to obtain is expressed in terms of $\boldsymbol{H}^H$, and this will simplify notation.}, $\boldsymbol{A}^H$, before being forwarded to the CPU. 
\begin{figure}[h]
	\centering
	\vspace*{-5mm}
	\includegraphics[scale=0.485]{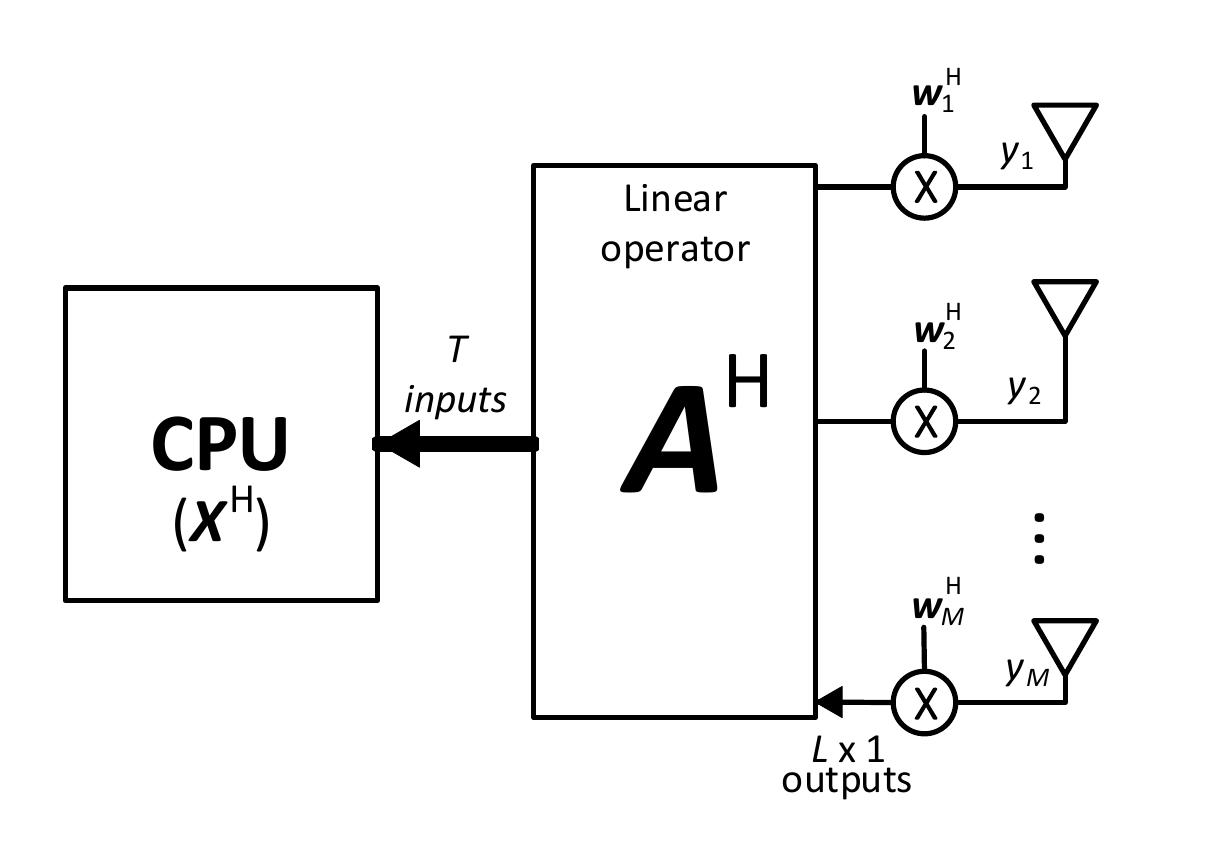}
	\vspace*{-4mm}
	\caption{Generalized architecture considered in this paper.}
	\label{fig:arch_new}
\end{figure}
Assuming that the CPU also applies a linear operator\footnotemark[1], $\boldsymbol{X}^H$, which corresponds to a $K \times T$ matrix, the post-processed signal within our generalized architecture can be expressed as
\begin{equation}
    \boldsymbol{z}=\boldsymbol{X}^H\boldsymbol{A}^H\boldsymbol{W}^H\boldsymbol{y},
\end{equation}
where $\boldsymbol{W}$ is an $M \times ML$ block diagonal matrix
\begin{equation}\label{eq:w}
    \boldsymbol{W}= \mathrm{diag}\left( \boldsymbol{w}_1, \boldsymbol{w}_2, \dots, \boldsymbol{w}_M \right).
\end{equation}
We should note that, within the defined architecture, we have absolute freedom in selecting $\boldsymbol{W}$ and $\boldsymbol{X}$, i.e., they can change from one time block to another, if $\boldsymbol{H}$ changes. However, $\boldsymbol{A}$ is seen as a system parameter that can be freely chosen when designing the system, but it is otherwise considered constant.

We are interested in applying a lossless transformation to the received signal, i.e., preserving the mutual information with respect to the transmitted signal $\boldsymbol{s}$, within our restricted architecture so that the architecture doesn't affect the performance of the system. As mentioned above, the MF is the simplest option to do so, and other equalization methods would still be possible after it in the CPU (assuming full CSI). Therefore, our architecture would be able to operate lossless if we can fulfill the equality
\begin{equation} \label{XAW}
    \boldsymbol{X}^H\boldsymbol{A}^H\boldsymbol{W}^H = \boldsymbol{H}^H.
\end{equation}
We will next investigate necessary conditions on the parameters $(M,K,T,L)$ and the matrix $\boldsymbol{A}$ for the equation \eqref{XAW}, with $\boldsymbol{W}$ and $\boldsymbol{X}$ as unknowns, to be solvable for any $\boldsymbol{H}$.
\section{WAX decomposition}
\label{section:wax}


Given an arbitrary $M \times K$ matrix, $\boldsymbol{H}$, and a fixed $M \times T$ matrix, $\boldsymbol{A}$, we define the WAX decomposition of $\boldsymbol{H}$ as
\begin{equation} \label{WAX}
  \boldsymbol{H} =  \boldsymbol{W}\boldsymbol{A}\boldsymbol{X}.
\end{equation}
This decomposition relates directly to (\ref{XAW}) taking conjugate transpose, i.e., $\boldsymbol{W}$ has the structure defined in (\ref{eq:w}) and $\boldsymbol{X}$ is a $T \times K$ matrix.

The generality of the WAX decomposition is established in the following theorem.
\begin{thm} \label{WAXmain}
Provided that 
\begin{equation}\label{eq:cond_wax}
T>\max\left(M\frac{K-L}{K},K-1\right)
\end{equation}
and with a randomly chosen $\boldsymbol{A}\in\mathbb{C}^{ML\times T}$, the set of matrices $\boldsymbol{H}$ that does not admit a decomposition of the form \eqref{WAX} has measure 0.
\end{thm}
\begin{proof}
See Appendix A.
\end{proof}

While Theorem \ref{WAXmain} states that any randomly chosen $\boldsymbol{A}$ works for a WAX decomposition, we are, from a practical perspective, interested in $\boldsymbol{A}$ matrices having simple forms (providing low computational complexity), for example, sparse matrices with elements in the set $\{0,1\}$. This would significantly simplify the combining network shown as $\boldsymbol{A}^H$ in Figure \ref{fig:arch_new}. However, for such a matrix, Theorem \ref{WAXmain} does no longer apply. As such, it is of importance to investigate the exceptions to Theorem \ref{WAXmain}. 

The full analysis of the exceptions is deferred to a subsequent paper, and in this paper, we will limit our investigation to a constrained form of $\boldsymbol{A}$. The form we consider is 
\begin{equation}\label{eq:A_block}
    \boldsymbol{A} = \widetilde{\mathbf{I}}_L \widetilde{\boldsymbol{A}},
\end{equation}
where $\widetilde{\boldsymbol{A}}$ is an arbitrary $M \times T$ matrix, and $\widetilde{\mathbf{I}}_L$ is a block matrix of the form
\begin{equation}
\widetilde{\mathbf{I}}_L = \mathbf{I}_N \otimes (\boldsymbol{1}_{L\times 1}\otimes \mathbf{I}_L),
\end{equation}
with $N=M/L$. Note that, for this structure to apply, $M/L$ has to evaluate to an integer value. This structure still allows for a sparse structure of $\boldsymbol{A}$ comprising only $\{0,1\}$ elements, and simplifies the analysis. Furthermore, the particular form would arise naturally in case the antennas are clustered into clusters of size $L$, as we elaborate more on shortly.
Substituting \eqref{eq:A_block} into \eqref{WAX} yields the decomposition
\begin{equation}\label{eq:wax2}
    \boldsymbol{H}=\widetilde{\boldsymbol{W}}\widetilde{\boldsymbol{A}}\boldsymbol{X},
\end{equation}
where $\widetilde{\boldsymbol{W}}$ is now an $M \times M$ matrix
\begin{equation}
    \widetilde{\boldsymbol{W}} = \boldsymbol{W} \widetilde{\mathbf{I}}_L,
\end{equation}
which corresponds to a block diagonal matrix of the form
\begin{equation}
    \widetilde{\boldsymbol{W}}=\mathrm{diag}\left( \widetilde{\boldsymbol{W}}_1,\widetilde{\boldsymbol{W}}_2,\dots,\widetilde{\boldsymbol{W}}_N \right),
\end{equation}
where $\widetilde{\boldsymbol{W}}_n$ are $L\times L$ matrices, and their rows  correspond to the vectors $\boldsymbol{w}_m$ from \eqref{eq:w}. If the antennas would be clustered into clusters of size $L$, then \eqref{eq:wax2} would appear naturally, which further motivates our limitation \eqref{eq:A_block}.

As happened in Theorem \ref{WAXmain} for the decomposition \eqref{WAX}, the set of $\boldsymbol{H}$ matrices that do not admit a decomposition of the form \eqref{eq:wax2} also has measure 0, if we have a randomly chosen $\widetilde{\boldsymbol{A}}$ and condition \eqref{eq:cond_wax} is fulfilled. We omit this proof for sake of space constraints, but mention that it follows the steps in the proof of Theorem \ref{WAXmain}, where in this case the constraint to be satisfied comes in the form of a determinant operation instead of a norm as in Lemma \ref{lemma_t1}.

Let us define $\boldsymbol{H}=[\boldsymbol{H}_1^T \, \boldsymbol{H}_2^T \, \dots \boldsymbol{H}_N^T]^T$ and $\widetilde{\boldsymbol{A}}=[\widetilde{\boldsymbol{A}}_1^T \, \widetilde{\boldsymbol{A}}_2^T \, \dots \widetilde{\boldsymbol{A}}_N^T]^T$, where $\boldsymbol{H}_n$ are $L \times K$ blocks and $\widetilde{\boldsymbol{A}}_n$ are $L \times T$ blocks. For practical computation of \eqref{eq:wax2}, the following lemma is useful.
\begin{lem} \label{lemeqv}
For all matrices $\boldsymbol{H}$ satisfying $\mathrm{rank}(\boldsymbol{H}_n)=L$ there exists a block diagonal matrix $\widetilde{\boldsymbol{W}}$ and a matrix $\boldsymbol{X}$ such that $\widetilde{\boldsymbol{W}}\widetilde{\boldsymbol{A}}\boldsymbol{X}=\boldsymbol{H}$, if and only if, there exists  a block diagonal invertible matrix $\widehat{\boldsymbol{W}}$ such that $\widetilde{\boldsymbol{A}}\boldsymbol{X}-\widehat{\boldsymbol{W}}\boldsymbol{H}=\boldsymbol{0}_{M\times K}$.
\end{lem}
\begin{proof}
Assume existence of block diagonal matrix $\widetilde{\boldsymbol{W}}$ and a matrix $\boldsymbol{X}$ such that $\widetilde{\boldsymbol{W}}\widetilde{\boldsymbol{A}}\boldsymbol{X}=\boldsymbol{H}$. This is equivalent to
$$\widetilde{\boldsymbol{W}}_n\widetilde{\boldsymbol{A}}_n\boldsymbol{X}=\boldsymbol{H}_n, \; 1\leq n \leq N.$$ 
Since, by assumption, $\mathrm{rank}(\boldsymbol{H}_n)=L$, it follows that $\mathrm{rank}(\widetilde{\boldsymbol{W}}_n)=L, \; \forall n$, making the matrix $\widetilde{\boldsymbol{W}}$ invertible.

The reverse statement is trivial; if an invertible $\widehat{\boldsymbol{W}}$ exists, then we can set $\widetilde{\boldsymbol{W}}=\widehat{\boldsymbol{W}}^{-1}.$
\end{proof}

For a randomly chosen $\boldsymbol{H}$, the condition $\mathrm{rank}(\boldsymbol{H}_n)=L$ holds with probability 1. We can then compute the WAX decomposition by invoking Lemma \ref{lemeqv}, which yields the linear system
\begin{equation}\label{eq:AXWH}
    \widetilde{\boldsymbol{A}}\boldsymbol{X}-\widetilde{\boldsymbol{W}}^{-1}\boldsymbol{H}=\boldsymbol{0}_{M\times K}.
\end{equation}
Using the vectorization operator we get the an equivalent linear system of equations
\begin{equation}\label{eq:lin_wax}
\boldsymbol{B}\boldsymbol{u}=\boldsymbol{0}_{MK\times 1},
\end{equation}
where $\boldsymbol{u}$ corresponds to the $(TK+ML)\times 1$ vector of unknowns,
\begin{equation}
\boldsymbol{u}=\begin{bmatrix}
\mathrm{vec}(\boldsymbol{X}) \\
\mathrm{vec}(\widetilde{\boldsymbol{W}}_1) \\
\vdots \\
\mathrm{vec}(\widetilde{\boldsymbol{W}}_N)
\end{bmatrix},
\end{equation}
and $\boldsymbol{B}$ is an $M K \times (TK+ML)$ block matrix of the form $\boldsymbol{B}=[\boldsymbol{B}_1 \;\boldsymbol{B}_2]$ resulting from the vectorization operation, with
\begin{equation} \label{eq:B1}
\boldsymbol{B}_1 = \mathbf{I}_K \otimes \widetilde{\boldsymbol{A}}, \quad \boldsymbol{B}_2 = -(\boldsymbol{H}^\mathrm{T} \otimes \mathbf{I}_M)\mathbf{P},
\end{equation}
where $\mathbf{P}$ is an $M^2\times ML$ block matrix composed of identity matrices, $\boldsymbol{I}_{L}$, separated by rows of zeros so as to disregard the zeros in $\mathrm{vec}(\boldsymbol{W})$. The solution to \eqref{eq:lin_wax} can be found by setting $\boldsymbol{u}$ to be any vector in the null-space of $\boldsymbol{B}$, which will always be non-zero if condition \eqref{eq:cond_wax} is met. Then we can obtain the corresponding $\widetilde{\boldsymbol{W}}^{-1}$ and $\boldsymbol{X}$ from $\boldsymbol{u}$ through inverse vectorization, and we should check that the resulting $\widetilde{\boldsymbol{W}}^{-1}$ is full rank so that we can obtain $\widetilde{\boldsymbol{W}}$ taking the matrix inverse. For a randomly chosen $\widetilde{\boldsymbol{A}}$, $\widetilde{\boldsymbol{W}}^{-1}$ is full rank with probability 1, but note that some $\widetilde{\boldsymbol{A}}$ matrices may lead to rank deficient $\widetilde{\boldsymbol{W}}^{-1}$ even if \eqref{eq:cond_wax} is met. That is, for a poorly chosen matrix $\widetilde{\boldsymbol{A}}$, a WAX decomposition of a matrix $\boldsymbol{H}$ cannot be performed. In what follows next, we study necessary conditions on $\widetilde{\boldsymbol{A}}$ in order for the WAX decomposition to be feasible.

\subsection{Conditions on matrix $\widetilde{\boldsymbol{A}}$}
We aim for a sparse $\boldsymbol{A}$ matrix with elements in the set $\{0,1\}$, instead of purely random as in Theorem \ref{WAXmain}. As earlier mentioned, the structure from \eqref{eq:A_block} maintains the sparsity properties of $\widetilde{\boldsymbol{A}}$, since it leads to an $\boldsymbol{A}$ matrix with repeated $\widetilde{\boldsymbol{A}}$ blocks. Therefore, we next study  conditions on $\widetilde{\boldsymbol{A}}$ for the decomposition \eqref{eq:wax2}, and thus \eqref{WAX}, to be possible.

\begin{defi}
We consider $\widetilde{\boldsymbol{A}}$ to be valid for the decomposition \eqref{eq:wax2} if the set of matrices $\boldsymbol{H}$ that does not admit such a decomposition has measure 0.
\end{defi}

It can be noticed that a valid $\widetilde{\boldsymbol{A}}$ leads also to a valid $\boldsymbol{A}$ for the decomposition \eqref{WAX} through \eqref{eq:A_block}.

\begin{lem} \label{lemma_blockrank}
If $\widetilde{\boldsymbol{A}}$ is valid, then
\begin{equation}\label{eq:rank_A}
    \mathrm{rank}(\widetilde{\boldsymbol{A}}_\ell) = L, \;\; \forall \ell.
\end{equation}
\end{lem}
\begin{proof}
For a valid $\widetilde{\boldsymbol{A}}$, we have
$\widetilde{\boldsymbol{W}}_\ell\widetilde{\boldsymbol{A}}_\ell\boldsymbol{X}=\boldsymbol{H}_\ell, \;\; \forall \ell.
$
For a random $\boldsymbol{H}$, $\mathrm{rank}(\boldsymbol{H}_\ell) = L, \;\; \forall \ell$
with probability 1. Since
\begin{equation}
    \mathrm{rank}(\widetilde{\boldsymbol{W}}_\ell\widetilde{\boldsymbol{A}}_\ell\boldsymbol{X})\leq \mathrm{rank}(\widetilde{\boldsymbol{W}}_\ell)\mathrm{rank}(\widetilde{\boldsymbol{A}}_\ell)\mathrm{rank}(\boldsymbol{X}),
\end{equation}
condition \eqref{eq:rank_A} must be fulfilled.
\end{proof}

A further necessary condition for $\widetilde{\boldsymbol{A}}$ to be valid is given in Theorem \ref{rankrow}.
\begin{thm} \label{rankrow}
Let $\widetilde{\boldsymbol{A}}_0$ be a submatrix of $\widetilde{\boldsymbol{A}}$ formed by selecting $R$ rows from $\widetilde{\boldsymbol{A}}$, where all rows are in different blocks $\widetilde{\boldsymbol{A}}_n$. If $\widetilde{\boldsymbol{A}}$ is valid, then
$$\mathrm{rank}(\widetilde{\boldsymbol{A}}_0) > R\frac{K-L}{K}$$
\end{thm}
\begin{proof}
See Appendix B.
\end{proof}

We point out that since $\widetilde{\boldsymbol{A}}$ is an $M\times T$ matrix, $\mathrm{rank}(\widetilde{\boldsymbol{A}}_0)$ cannot exceed $T$. However, with $T>M(K-L)/K$, it is guaranteed that $R\frac{K-L}{K}<T$.  

An immediate consequence of Theorem \ref{rankrow} is that a block $\boldsymbol{A}_\ell$ cannot be repeated arbitrarily often. In fact, a block $\widetilde{\boldsymbol{A}}_\ell$ must have, from
Lemma \ref{lemma_blockrank}, rank $L$. Repeating the block $\widetilde{\boldsymbol{A}}_\ell$ $r$ times in $\widetilde{\boldsymbol{A}}$, and selecting $\widetilde{\boldsymbol{A}}_0$ as the same row within each of these $r$ blocks yields,
$$1>r\frac{K-L}{K},$$
which implies $r<\frac{K}{K-L}.$  Whenever $L\leq K/2$, $r=1$ so that each block $\widetilde{\boldsymbol{A}}_\ell$ can only occur once in $\widetilde{\boldsymbol{A}}$.

\section{Discussion and Examples}
\label{section:tradeoff}
As mentioned previously, one of the goals of this paper is to find the trade-off between the number of outputs per antenna, $L$, and the number of inputs to the CPU, $T$, required for the overall receiver to be lossless. Assuming that our generalized architecture is equipped with a valid matrix $\boldsymbol{A}$, the trade-off between $L$ and $T$ comes directly from condition \eqref{eq:cond_wax}. To elaborate a bit further, we observe that with $L=1$, the number of CPU inputs becomes $T=T_\mathrm{max} \triangleq \left\lfloor M-\frac{M}{K}+1\right \rfloor$. 
In general, we can select $T$ as
$$T=\max\left(\left\lfloor M\frac{K-L}{K}+1\right \rfloor,K\right)$$
which is conceptually shown in the left part of Figure \ref{fig:tradeoff}.

It is interesting to observe that we reach a reduction compared with the centralized architecture also for $L=1$. This reduction comes about since we have allowed the antennas to perform multiplications, which leads to a reduction in the number of CPU inputs from $M$ to, at most, $T_\mathrm{max}$. The centralized architecture, illustrated in the left part of Figure 1, has the same number of outputs per antenna, namely $1$, but does not perform any multiplications. Therefore, the CPU must operate with $T=M$. 
If we let $L_{\mathrm{mult}}$ denote the number of multiplications per antenna, the centralized architecture corresponds to $L_{\mathrm{mult}}=0$, and we can select $T$ as
$$T= \left\{ \begin{array}{ll} M & L_{\mathrm{mult}} =0 \\ \max\left(\left\lfloor M\frac{K-L_{\mathrm{mult}}}{K}+1\right \rfloor,K\right) & L_{\mathrm{mult}} >0. \end{array} \right.$$
This is conceptually illustrated in the right part of Figure \ref{fig:tradeoff}.
\begin{figure}
     \centering
     \begin{subfigure}[b]{0.22\textwidth}
         \centering
         \includegraphics[scale=0.85]{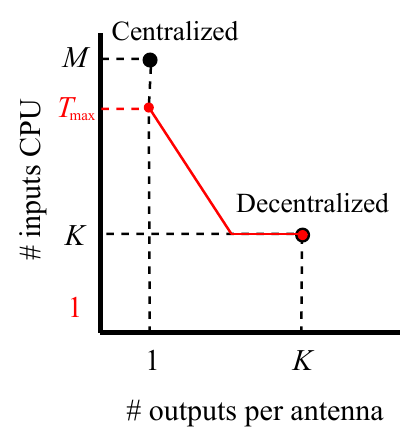}
     \end{subfigure}
     \hfill
     \begin{subfigure}[b]{0.22\textwidth}
         \centering
         \includegraphics[scale=0.85]{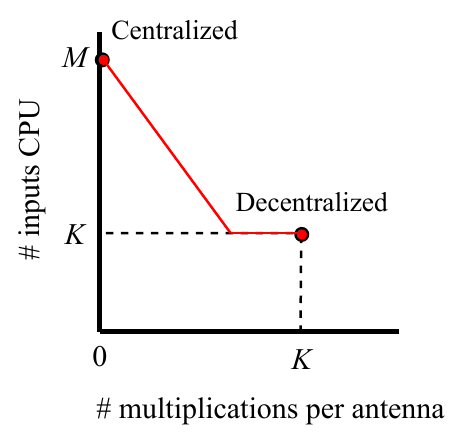}
     \end{subfigure}
        \caption{Number of inputs to the CPU v.s. number of outputs (left)/number of multiplications (right) per antenna.}
        \label{fig:tradeoff}
\end{figure}

We conclude this section with a few examples.
\begin{ex} 
Assume a design with a CPU limited to $T\leq 50$ inputs, and antenna modules with $L=2$ outputs. We now consider how many antennas $M$ and users $K$ the system can handle. From \eqref{eq:cond_wax}, we have  $50>M\frac{K-2}{K}$ implying that
$$50\frac{K}{K-2}>M.$$
To maximize the left hand side, save for the special case $K=L=2$ which allows for an unlimited number of antennas, we set $K=3$ and obtain $M<150$, so that we can at most use $149$ antennas. Or, put differently, if we choose to equip the base station with 149 antennas, we can at most serve $K=3$ users. With 150+ antennas, only 2 users can be served. Setting $K=4$, yields that at most 99 antennas can be used.
\end{ex}

We next provide two numerical examples of the WAX decomposition. The first one is meant to illustrate that it is indeed possible to obtain valid sparse matrices $\boldsymbol{A}$ and $\widetilde{\boldsymbol{A}}$ comprising only elements in the set $\{0,1\}$.
\begin{ex} Let $M=100$, $K=10$, and $L=4$. From Theorem 1, we have that $T>100\times 0.6=60$, so we take $T=61$. It can be numerically verified that the matrix
\begin{equation} \label{ex1} 
\widetilde{\boldsymbol{A}} = \left[\begin{array}{c} 
\boldsymbol{I}_{61} \\
\hline
\begin{array}{c|c}
\boldsymbol{I}_{39} & \begin{array}{c} \boldsymbol{I}_{22} \\ \hline \begin{array}{c|c} 
\boldsymbol{I}_{17} & \begin{array}{c} \boldsymbol{I}_5 \\ \boldsymbol{I}_5 \\ \boldsymbol{I}_5 \\ \hline \begin{array}{ccc} \boldsymbol{I}_2 & \boldsymbol{I}_2 & \boldsymbol{0}_{2\times 1} \end{array}
\end{array}
\end{array} 
\end{array}
\end{array}
\end{array}
\right]\end{equation}
is valid. We designed this $\widetilde{\boldsymbol{A}}$ by aiming at a minimum number of non-zero elements, while satisfying both Lemma \ref{lemma_blockrank} and Theorem \ref{rankrow}. The matrix $\boldsymbol{A}$ is obtained as $\boldsymbol{A}=\tilde{\boldsymbol{I}}_L\widetilde{\boldsymbol{A}}$, and it can be verified that $\boldsymbol{A}$ has 612 ones and 23788 zeros. Thus, merely 2.5\% of $\boldsymbol{A}$ is non-zero. 
\end{ex}

Our next example is providing the reader with a graphical illustration of the WAX decomposition.
\begin{ex} 
Let $M=8$, $K=5$, and $L=2$. Thus, $T>4.8$, so we select $T=5$. In this case, the number of variables is $TK+ML=41$ and the number of equations is $MK=40$; thus, we have precisely one more variable than equations. A particular example of the WAX decomposition is shown in \eqref{ex2} for the block version using $\widetilde{\boldsymbol{W}}$ and $\widetilde{\boldsymbol{A}}$ (expressed with $\boldsymbol{W}$ and $\boldsymbol{A}$, the $\boldsymbol{W}$ matrix would be twice as wide, and $\boldsymbol{A}$ twice as tall). The strength of the WAX decomposition is that for any $\boldsymbol{H}$, except for those in a set of measure 0, the matrix $\widetilde{\boldsymbol{A}}$ can be kept as it is, while only $\widetilde{\boldsymbol{W}}$ and $\boldsymbol{X}$ need to change.
\begin{figure*}[b]
\hrulefill
$$ $$
\begin{equation} \label{ex2}
\underbrace{\left[\begin{BMAT}{cccccccc}{cccccccc} \,1&-1&0&0&0&0&0&0 \\
                                          \,2&1&0&0&0&0&0&0 \\
                                       \,   0&0&1&1&0&0&0&0 \\
                                       \,   0&0&0&-1&0&0&0&0 \\
                                      \,    0&0&0&0&2&1&0&0 \\
                                      \,    0&0&0&0&2&-1&0&0 \\
                                      \,    0&0&0&0&0&0&2&-1 \\
                                      \,    0&0&0&0&0&0&2&2
\addpath{(0,6,1)rruulldd}
\addpath{(2,4,1)rruulldd}
\addpath{(4,2,1)rruulldd}
\addpath{(6,0,1)rruulldd}
\end{BMAT} \right]}_{\widetilde{\boldsymbol{W}}}
\underbrace{\left[ \begin{BMAT}{ccccc}{cccccccc}1&0&0&0&0 \\
                                          0&1&0&0&0 \\
                                          0&0&1&0&0 \\
                                          0&0&0&1&0 \\
                                          0&0&0&0&1 \\
                                          1&0&0&1&1 \\
                                          0&1&0&1&1 \\
                                          0&0&1&1&1
\end{BMAT} \right]}_{\widetilde{\boldsymbol{A}}} 
\underbrace{\left[ \begin{BMAT}{ccccc}{ccccc} 
-2&-1&-1&2&2 \\
1&-2&-1&1&2 \\
1&-1&-2&-1&-2 \\
0&2&0&-1&0 \\
0&-1&2&1&-2  
\end{BMAT} \right]}_{\boldsymbol{X}} =
\underbrace{\left[ \begin{BMAT}{ccccc}{cccccccc} 
-3&1&0&1&0 \\
-2&-4&-3&5&6 \\
1&1&-2&-2&-2 \\
0&-2&0&1&0 \\
-2&-2&5&4&-4 \\
2&-2& 3&0&-4 \\
1&-2&2&3&4 \\
4&-2&2&0&-8 
\end{BMAT} \right]}_{\boldsymbol{H}}
\end{equation}
\end{figure*}
\end{ex}

%
%

\section{Conclusions}
\label{section:conclusions}
We have presented a generalized architecture that allows to exploit the trade-off between centralized and decentralized massive MIMO architectures,  in terms of number of connections to the CPU and number of outputs per antenna. We have characterized said trade-off by defining a new matrix decomposition that allows for lossless linear equalization within our architecture, and deriving the conditions for it to work lossless.

Future work could include a characterization of the incurred losses if the number of CPU inputs goes below the limit for lossless processing.
\section*{Appendix A: Proof of Theorem \ref{WAXmain}}
We first make the observation that the rank of $\boldsymbol{A}$ can be no less than the rank of $\boldsymbol{H}$. Assuming that $M\geq K$, this implies  $T\geq K$, expressed as $T> \max(\cdot, K-1)$ in the statement.

We next provide  a lemma that will be useful.
\begin{lem} \label{lemma_t1}
Let $\boldsymbol{W}$ and $\bar{\boldsymbol{W}}$ be two matrices of the same form as $\boldsymbol{W}$ in \eqref{WAX}. Then, if $\boldsymbol{AX}=\bar{\boldsymbol{W}}^H \boldsymbol{H}$ is solvable such that $\|\bar{\boldsymbol{w}}_m\|^2 > 0, \; 1\leq m \leq M$, then $\boldsymbol{W}\boldsymbol{A}\boldsymbol{X}=\boldsymbol{H}$ is solvable.
\end{lem}
\begin{proof}
Suppose $\boldsymbol{A}\boldsymbol{X}=\bar{\boldsymbol{W}}^H \boldsymbol{H}$ is solvable such that $\|\bar{\boldsymbol{w}}_m\|^2 > 0, \; 1\leq m \leq M$. This implies that $\bar{\boldsymbol{W}}^H$ has left inverse $(\bar{\boldsymbol{W}}\bar{\boldsymbol{W}}^H)^{-1}\bar{\boldsymbol{W}}$. Thus,
$$(\bar{\boldsymbol{W}}\bar{\boldsymbol{W}}^H)^{-1}\bar{\boldsymbol{W}}{\boldsymbol{A}}{\boldsymbol{X}}=\boldsymbol{H}.$$
The lemma follows by observing that $(\bar{\boldsymbol{W}}\bar{\boldsymbol{W}}^H)^{-1}\bar{\boldsymbol{W}}$ is of the same form as $\boldsymbol{W}$, so we can take $\boldsymbol{W}=(\bar{\boldsymbol{W}}\bar{\boldsymbol{W}}^H)^{-1}\bar{\boldsymbol{W}}$.
\end{proof}

Let us now study $\boldsymbol{AX}=\bar{\boldsymbol{W}}^H \boldsymbol{H}$. Said matrix equation specifies $MT$ linear equations in $TK+ML$ variables, and, hence, solvable if $T>M(K-L)/K$. It remains to show that for randomly chosen $\boldsymbol{A}$ and $\boldsymbol{H}$, the solution satisfies $\|\bar{\boldsymbol{w}}_m\|^2>0, \, 1\leq m\leq M$. Let us define $\mathcal{V}$ as the set of admissible solutions, i.e., 
$$\mathcal{V}=\{\boldsymbol{A}, \boldsymbol{H} \, | \, \exists \bar{\boldsymbol{W}}^H, \boldsymbol{X}: \boldsymbol{AX}=\bar{\boldsymbol{W}}^H \boldsymbol{H}, \|\bar{\boldsymbol{w}}_m\|^2>0, \, \forall m\}$$

Assume the contrary, i.e., that there exists an $m_0$ such that $\|\bar{\boldsymbol{w}}_{m_0}\|^2=0$. The solution to $\boldsymbol{AX}=\bar{\boldsymbol{W}}^H \boldsymbol{H}$ depends on $S=TK+ML-MK$ free variables, here denoted by $\{z_s\}$. 
The solution $\bar{w}_{m_0,\ell}$ is a linear combination of the free variables $\{z_s\}$ where the weights depend on $\boldsymbol{A}$ and $\boldsymbol{H}$, i.e.,
$$\bar{w}_{m_0,\ell}=\sum_{s=1}^S c_{m_0,\ell,s}(\boldsymbol{A},\boldsymbol{H})z_s.$$
Now, invoking the assumption that $\|\bar{\boldsymbol{w}}_{m_0}\|^2=0$, implies that $\bar{w}_{m_0,\ell}=0, \, 1\leq \ell \leq L$ no matter the values of $\{z_s\}$. The only possibility for this is if the coefficients $c_{m_0,\ell,s}(\boldsymbol{A},\boldsymbol{H})=0, \, 1\leq \ell \leq L, \, 1\leq s \leq S$.

However, the coefficients $c_{m_0,\ell,s}(\boldsymbol{A},\boldsymbol{H})$ are rational expressions of the entries in $\boldsymbol{A}$ and $\boldsymbol{H}$, which means that in order for $c_{m_0,\ell,s}(\boldsymbol{A},\boldsymbol{H})=0$, a polynomial multi-variate expression in the entries in $\boldsymbol{A}$ and $\boldsymbol{H}$ must be 0. A standard result from Zariski topology states that whenever $\mathcal{V}\neq \emptyset$, the set $\mathcal{V}$ is open dense in $\mathbb{C}^{TK+ML}$. This, in turn, implies that a randomly chosen tuple $(\boldsymbol{A}, \boldsymbol{H})$ is in $\mathcal{V}$ with probability 1.

\section*{Appendix B: Proof of Theorem \ref{rankrow}}
From the structure of $\boldsymbol{B}_1$ in \eqref{eq:B1}, we observe that a particular row of $\widetilde{\boldsymbol{A}}$ appears exactly in $K$ rows of $\boldsymbol{B}$. Let the submatrix of $\boldsymbol{B}$ formed by all rows in $\boldsymbol{B}$ where the rows of $\widetilde{\boldsymbol{A}}_0$ appear, be denoted as $\boldsymbol{B}_0$. Clearly, to satisfy \eqref{eq:lin_wax}, we must in particular satisfy $\boldsymbol{B}_0 \boldsymbol{u}=\boldsymbol{0}_{RK\times 1}$. Now, $\boldsymbol{B}_0$ reads
$$\boldsymbol{B}_0=\left[\boldsymbol{I}_K \otimes \widetilde{\boldsymbol{A}}_0 \;\;  \widehat{\boldsymbol{H}}_0\right],$$
where $\widehat{\boldsymbol{H}}_0$ is formed from $\boldsymbol{H}$ as follows: Let $\iota(r)$ denote the block $\boldsymbol{H}_{\iota(r)}$ where the $r$th row in $\widetilde{\boldsymbol{A}_0}$ is taken from. Let $\boldsymbol{H}_0=\left[ \boldsymbol{H}_{\iota(1)}^T \; \boldsymbol{H}_{\iota(2)^T} \, \dots \, \boldsymbol{H}_{\iota(R)}^T \right]^T$, and let $\mathbb{I}\mathbb{I}(\ell)$ be an $R\times L$ matrix  with a single entry equal to 1 at row $\ell$ and column $(\iota(\ell) \;\mathrm{mod} \; L)+1$, and all other equal to 0. Then,
\begin{eqnarray} \label{hatH} \widehat{\boldsymbol{H}}_0&=&\left[\boldsymbol{0}_{\mathcal{D}_0} \;\; \boldsymbol{H}_{\iota(1)}^H \!\otimes\! \mathbb{I}\mathbb{I}(1) \; \; \boldsymbol{0}_{\mathcal{D}_1}  \;\; \;\boldsymbol{H}_{\iota(2)}^H \!\otimes\! \mathbb{I}\mathbb{I}(2) \; \right. \nonumber \\
&& \;\quad \quad \quad \left. \dots \quad \boldsymbol{0}_{\mathcal{D}_{R-1}}  \;\; \boldsymbol{H}_{\iota(R)}^H\! \otimes \!\mathbb{I}\mathbb{I}(R)  \;\; \boldsymbol{0}_{\mathcal{D}_R} \right]\end{eqnarray}
where we have used the shorthand notation 
\begin{eqnarray} 
\mathcal{D}_k&=&RK\times (\iota(k+1)-\iota(k))L^2, \nonumber \\
&& \iota(0)\triangleq 1,  \; \iota(R+1)\triangleq M/L. \nonumber 
\end{eqnarray}

To study the null space of $\boldsymbol{B}_0$ we may just as well study the null space of $(\boldsymbol{I}_K \otimes \boldsymbol{Q}_0^H)\boldsymbol{B}_0$, where $\boldsymbol{Q}_0\boldsymbol{R}_0 = \widetilde{\boldsymbol{A}}_0$ is the QR decomposition of $\widetilde{\boldsymbol{A}}_0$. We have, 
\begin{equation} \label{appb:qr}
(\boldsymbol{I}_K \otimes \boldsymbol{Q}_0^H)\boldsymbol{B}_0= \left[\boldsymbol{I}_K \otimes {\boldsymbol{R}}_0 \; \; (\boldsymbol{I}_K \otimes \boldsymbol{Q}_0^H)\widehat{\boldsymbol{H}}_0\right].\end{equation} Let $\kappa= \mathrm{rank}(\widetilde{\boldsymbol{A}}_0)$. The matrix $\boldsymbol{I}_K \otimes {\boldsymbol{R}}_0$ consequently has $K(R-\kappa)$ all-zero rows. 

If we extract said all-zero rows, we obtain,$$
\left[\boldsymbol{0}_{K\!(\!R-\kappa)\times TK} \; \boldsymbol{P}(\boldsymbol{I}_K \otimes \boldsymbol{Q}_0^H)\widehat{\boldsymbol{H}}_0\right] \!\!\begin{bmatrix}
\mathrm{vec}({\boldsymbol{X}})\\
\mathrm{vec}(\widetilde{\boldsymbol{W}}_1)\\
\vdots \\
\mathrm{vec}(\widetilde{\boldsymbol{W}}_N)
\end{bmatrix}\!=\! \boldsymbol{0}_{K\!(\!R-\kappa)\times 1} 
$$
where $\boldsymbol{P}$ is an $K(R-\kappa)\times KR$ matrix that extracts the rows where $\boldsymbol{I}_K \otimes {\boldsymbol{R}}_0$ is all-zero. This implies that we can discard $\boldsymbol{X}$ so that we, equivalently, obtain
\begin{equation} \label{imop2}  \boldsymbol{P}(\boldsymbol{I}_K \otimes \boldsymbol{Q}_0^H)\widehat{\boldsymbol{H}}_0\begin{bmatrix}
\mathrm{vec}(\widetilde{\boldsymbol{W}}_1)\\
\vdots \\
\mathrm{vec}(\widetilde{\boldsymbol{W}}_N)
\end{bmatrix} = \boldsymbol{0}_{K(R-\kappa)\times 1}.\end{equation}

We next note that due to the many all-zero columns in $\widehat{\boldsymbol{H}}_0$, represented by the terms $\boldsymbol{0}_{\mathcal{D}_k}$ in \eqref{hatH}, not all the $\boldsymbol{W}_n$ matrices matter. In fact, it can be straightforwardly verified that \eqref{imop2} is equivalent to
\begin{equation} \label{imop3} \boldsymbol{P}(\boldsymbol{I}_K \otimes \boldsymbol{Q}_0^H)\bar{\boldsymbol{H}}_0 \begin{bmatrix}
\widetilde{\boldsymbol{w}}_{\iota(1)} \\
\vdots \\
\widetilde{\boldsymbol{w}}_{\iota(R)} 
\end{bmatrix} = \boldsymbol{0}_{K(R-\kappa)\times ML},\end{equation}
where $\widetilde{\boldsymbol{w}}_{m}$ is the $1\times L$ vector formed from extracting the entries at the $m$th row of $\widetilde{\boldsymbol{W}}$ that are allowed to take non-zero values, and
\begin{equation} \label{barH} \bar{\boldsymbol{H}}_0=\left[ \boldsymbol{H}_{\iota(1)}^H \!\otimes\! \mathbb{I}(1) \;  \;\boldsymbol{H}_{\iota(2)}^H \!\otimes\! \mathbb{I}(2)  \;\dots \; \boldsymbol{H}_{\iota(R)}^H\! \otimes \!\mathbb{I}(R) \right],\end{equation}
where $\mathbb{I}(\ell)$ is the non-zero column of $\mathbb{I}\mathbb{I}(\ell)$.

For randomly chosen $\boldsymbol{H}$, the matrix $\boldsymbol{P}(\boldsymbol{I}_K \otimes \boldsymbol{Q}_0^H)\bar{\boldsymbol{H}}_0$ is full rank with probability 1. Therefore, \eqref{barH} only has a non-trivial solution whenever the number of unknowns is larger than the number of equations, i.e., whenever, $RL > K(R-\kappa)$. Consequently, 
$$\kappa > R\frac{K-L}{K}.$$ 

\bibliographystyle{IEEEtran}
\bibliography{IEEEabrv,wax}

\end{document}